\documentclass[11pt]{article}

\usepackage[ruled, vlined, linesnumbered, noresetcount]{algorithm2e}
\usepackage{times}
\usepackage{epsfig}
\usepackage{amsthm}
\usepackage{amsmath}
\usepackage{amsfonts}
\usepackage{amssymb}
\usepackage{enumerate}
\usepackage{graphics}
\usepackage{color}
\usepackage{authblk}

\SetKwFor{ForAll}{forall}{do}{end forall}

\newtheorem{theorem}{Theorem}
\newtheorem{lemma}{Lemma}
\newtheorem{prop}{Proposition}
\newtheorem{definition}{Definition}

\newtheorem{claim}{Claim}
\newtheorem{remk}{Remark}
\newtheorem{proper}{Property}

\newcommand{\remove}[1]{}

\newcommand{\tm}[1]{round(#1)}
\newcommand{\pth}[1]{path(#1)}
\newcommand{\A}[1]{Act(#1)}

\newcommand{\Acttt}[3]{{\sf Active}[#1,#2,#3]}
\newcommand{\Actt}[2]{{\sf Active}[#1,#2]}

\newcommand{\tpth}[1]{maxPath(#1)}
\newcommand{\tpths}[2]{maxPath_{#2}(#1)}

\newcommand{\tms}[2]{round_{#2}(#1)}
\newcommand{\pths}[2]{path_{#2}(#1)}

\newcommand{\N}{{\mathbb{N}}}
\newcommand{\Active}{{{\sf Active}}}

\newcommand*{\cw}{\mathrm{cw}}

\begin{document}
\title{Latency-Bounded Target Set Selection in Social Networks\thanks{An extended abstract of this paper will appear in 
Proceedings of Computability in Europe 2013 (CiE 2013), 
The Nature of Computation: Logic, Algorithms, Applications, 
Lectures Notes in Computer Science, Springer.}}

%
%

\author[1]{F. Cicalese}

\author[2]{G. Cordasco}

\author[1]{L. Gargano}

\author[3]{M. Milani\v{c}}

\author[1]{U. Vaccaro}

\affil[1]{Dept.~ of Computer Science, University of Salerno, Italy, {\texttt{\{cicalese,lg,uv\}@dia.unisa.it}}}
\affil[2]{Dept.~of Psychology, Second University of Naples, Italy, {\texttt {gennaro.cordasco@unina2.it}}}
\affil[3]{University of Primorska, UP IAM and UP FAMNIT, 
SI 6000 Koper, Slovenia,  \texttt{martin.milanic@upr.si}}


\maketitle
\begin{abstract}
Motivated by applications in sociology, economy and medicine,
 we  study variants of the Target Set Selection problem,
first proposed by Kempe, Kleinberg and Tardos.
In our scenario one is given a graph $G=(V,E)$,
integer values $t(v)$ for each vertex $v$ (\emph{thresholds}), and the objective is to determine
a small set of vertices
(\emph{target set}) that activates a given number
(or a given subset) of vertices of $G$ \emph{within} a prescribed
number of rounds. The activation process in $G$ proceeds as follows: initially,
at round 0, all vertices in the target set
are activated;
subsequently at each round $r\geq 1$ every vertex
of $G$ becomes activated if at least $t(v)$ of its neighbors
are already active by round $r-1$.
It is known that the problem of finding a minimum cardinality
Target Set that eventually activates the whole graph $G$ is hard to approximate to a factor better than
 $O(2^{\log^{1-\epsilon }|V|})$.
In this paper we give \emph{exact} polynomial time algorithms
to find minimum cardinality
Target Sets in
 graphs of bounded clique-width, and \emph{exact}
linear time algorithms for trees.
%
\end{abstract}

\section{Introduction}
Let $G = (V,E)$ be a graph, $S \subseteq V$, and let $t: V \longrightarrow \N = \{1,2,\ldots\}$ be a
function assigning integer thresholds to the vertices of $G$. An {\em activation process in $G$ starting at $S$}
is a sequence  $\Active[S,0] \subseteq \Active[S,1] \subseteq \ldots\subseteq \Active[S,i] \subseteq \ldots \subseteq V$
of vertex subsets, with
$\Active[S,0] = S$, and such that for all $i > 0$,
$$\Active[S,i] = \Active[S,i-1]\cup \Big\{u \,:\, \big|N(u)\cap \Active[S,i - 1]\big|\ge t(u)\Big\}\,$$
where $N(u)$ is the set of neighbors of $u$.
In words, at each round  $i$ the set of active nodes is
augmented by the set of  nodes $u$ that have a number of
\emph{already} activated neighbors greater or equal to
 $u$'s threshold $t(u)$.
The central problem
we introduce and study in this paper
 is defined as follows:

\medskip
\noindent {\sc $(\lambda, \beta, \alpha)$-Target Set Selection ($(\lambda, \beta, \alpha)$-TSS)}.\\
{\bf Instance:} A graph $G=(V,E)$, thresholds $t:V\longrightarrow \mathbb{N}$, a latency bound $\lambda\in \N$,
a budget $\beta \in \N$ and an activation requirement $\alpha\in \N$.\\
{\bf Problem:} Find  $S\subseteq V$ s.t.~$|S|\le \beta$ and $|\Active[S,\lambda]|\ge \alpha$ (or determine that no such a set exists).

%

We will be also interested in the case in which {\em a set of nodes that need to be activated} (within the given latency bound)
is explicitly given as part of the input.

\medskip
\noindent {\sc $(\lambda, \beta, A)$-Target Set Selection ($(\lambda, \beta, A)$-TSS)}.\\
{\bf Instance:} A graph $G=(V,E)$, thresholds $t:V\longrightarrow \mathbb{N}$, a latency bound $\lambda\in \N$,
a budget $\beta \in \N$ and a set to be activated $A\subseteq V$.\\
{\bf Problem:} Find a set $S\subseteq V$ such that $|S|\le \beta$ and
{$A\subseteq \Active[S,\lambda]$} (or determine that such a set does not exist).

{Eliminating} any one of the parameters $\lambda$ and $\beta$, one  obtains
 two natural minimization problems.
For instance, eliminating  $\beta$, one  obtains the following problem:

\medskip
\noindent {\sc $(\lambda, A)$-Target Set Selection ($(\lambda, A)$-TSS)}.\\
{\bf Instance:} A graph $G=(V,E)$, thresholds $t:V\longrightarrow \mathbb{N}$, a latency bound $\lambda\in \N$
and a set 
$A\subseteq V$.\\
{\bf Problem:} Find a set $S\subseteq V$ of minimum size such that
{$A\subseteq \Active[S,\lambda]$}.

\begin{sloppypar}
{Notice that in the above problems we may assume without loss of generality that $0{\le} t(u){\le} d(u){+}1$ holds for all nodes $u{\in} V$ (otherwise,
we can set $t(u) {=} d(u){+}1$ for every node $u$ with threshold exceeding its degree plus one without changing the problem).}
\end{sloppypar}

The above algorithmic problems have roots in the general study
of the \emph{spread of influence}
in Social Networks (see  \cite{EK} and references quoted therein).
For instance, in the area of viral marketing \cite{DR-01,DNT12} companies  wanting to
promote products or behaviors might try initially  to target and convince
a few individuals which, by word-of-mouth effects, can  trigger
a  cascade of influence in the network, leading to
an  adoption  of the products by  a much larger number of individuals.
It is clear that the $(\lambda, \beta, \alpha)$-TSS problem represents
an abstraction of that scenario, once one makes the reasonable assumption  that an individual
decides to adopt the products if a certain number of {his/her friends} have adopted
said products. Analogously, the $(\lambda, \beta, \alpha)$-TSS problem can describe  other
diffusion problems  arising in sociological, economical  and biological  networks,
again see   \cite{EK}.
Therefore,  it comes as no surprise that
{special  cases of our} problem (or  variants thereof) have recently
attracted much attention by the algorithmic community.
In this version of the paper  we shall limit ourselves  to discuss the work which is
strictly  related to the present paper
(we just mention that our results are also relevant to other  areas,
like   dynamic monopolies \cite{FKRRS-2003,Peleg-02}, for instance).
The first authors to study problems of spread of influence in networks
from an algorithmic point of view were Kempe \emph{et al.} \cite{KKT-03,KKT-05}.
However, they were mostly interested in networks with  randomly chosen thresholds.
Chen \cite{Chen-09} studied the following minimization problem:
Given a graph $G$ and fixed thresholds $t(v)$, find
a target set of minimum size that eventually activates
all (or a fixed fraction of) vertices of $G$.
He proved  a  strong inapproximability result that makes unlikely the existence
of an  algorithm with  approximation factor better than  $O(2^{\log^{1-\epsilon }|V|})$.
Chen's result stimulated the work \cite{ABW-10,BHLM-11,Chopin-12}.
In particular, in \cite{BHLM-11}, Ben-Zwi {\em et al.}
proved that the $(|V|, \beta, \alpha)$-TSS problem can be
solved in time $O(t^w|V|)$ where $t$ is the maximum threshold
and $w$ is the treewidth of the graph, thus  showing that this variant of the problem is fixed-parameter
tractable if parameterized w.r.t. both treewidth and the maximum degree
of the graph.
Paper \cite{Chopin-12} isolated other interesting cases in which
the problems become efficiently tractable.

{
All the above mentioned papers did not consider the issue of
the number of rounds necessary for the activation of the required
number of vertices. However, this is a relevant question: In viral marketing,
for instance, it is quite important to spread information quickly.
It is equally important, before embarking on a possible
onerous investment,
to try  estimating the maximum amount of influence spread that can
be guaranteed within a certain amount of time (i.e, for
some $\lambda$ fixed in advance),
rather than simply knowing that eventually (but maybe too late)
the whole market might be covered. These considerations motivate our first generalization of
the problem, parameterized on the number of rounds $\lambda.$
The practical relevance of parameterizing the problem also
with  bounds on the initial budget or the final requirement should
be equally evident.

For general graphs, Chen's  \cite{Chen-09} inapproximability result
still holds if one demands that the activation process ends in a
bounded number of rounds.
We show that  the general $(\lambda, \beta, \alpha)$-TSS problem
 is polynomially solvable in graph of bounded clique-width and constant
 latency bound $\lambda$ (see Theorem \ref{thm1} in Section \ref{sec:cwd}).
 Since graphs of bounded treewidth are also of bounded clique-width \cite{CR05},
 this result implies a polynomial solution of the $(\lambda,\beta,\alpha)$-TSS
 problem with constant $\lambda$ also for graphs of bounded treewidth,
 complementing the result of \cite{BHLM-11} showing that for bounded-treewidth graphs,
 the TSS problem without the latency bound (equivalently, with $\lambda = |V|-1$) is polynomially solvable.
Moreover, the result settles the status of the computational complexity of the {\sc Vector Domination} problem
for graphs of bounded tree- or clique-width, that was posed as an open question in~\cite{CMV12+}.

We also consider the instance  when $G$ is a tree.  For this special case we
give  an \emph{exact linear time} algorithm for the $(\lambda, A)$-TSS  problem,
\emph{for any}  $\lambda$ and $A \subseteq V$.  When  $\lambda = |V|-1$ and $A = V$
our result is equivalent  to the (optimal)
linear time
algorithm  for the classical TSS problem (i.e., without the latency bound) on trees proposed in \cite{Chen-09}.

}
\section{TSS Problems on Bounded Clique-Width Graphs}\label{sec:cwd} 

In this section, we give an algorithm for the
{\sc $(\lambda, \beta, \alpha)$-Target Set Selection} problem
on graphs $G$ of clique-width at most $k$ given by an irredundant $k$-expression $\sigma$.
For the sake of self-containment we recall here some basic notions about clique-width.

\medskip
\noindent
{\bf The clique-width of a graph.} A {\em labeled graph} is a graph in which every vertex has a label from $\mathbb N$. A labeled graph is a {\em $k$-labeled graph} if every label is from $[k]:=\{1,2,\ldots,k\}$.
The {\em clique-width} of a graph $G$ is the minimum number of labels needed to construct $G$ using
the following four operations:
(i) Creation of a new vertex $v$ with label $a$ (denoted by $a(v)$);
(ii) disjoint union of two labeled graphs $G$ and $H$ (denoted by $G\oplus H$);
(iii) Joining by an edge each vertex with label $a$ to each vertex with label $b$
($a\not= b$, denoted by $\eta_{a,b}$);
(iv) renaming label $a$ to $b$
(denoted by $\rho_{a\to b}$).
Every graph can be defined by an algebraic expression using these four operations.
For instance, a chordless path on five consecutive vertices $u,v,x,y,z$ can be defined as follows:\\
$\eta_{3,2}(3(z)\oplus\rho_{3\to 2}(\rho_{2\to 1}(\eta_{3,2}
(3(y)\oplus\rho_{3\to 2}(\rho_{2\to 1}(\eta_{3,2}(3(x)\oplus\eta_{2,1}(2(v)\oplus 1(u))))))))).
$\\
Such an expression is called a $k$-{\em expression} if it uses at most $k$ different labels. The clique-width of $G$, denoted
$\cw(G)$, is the minimum $k$ for which there exists a $k$-expression defining $G$.
If a graph $G$ has a clique-width at most $k$, then a $(2^{k+1}- 1)$-expression for it can be computed in time $O(|V(G)|^3)$
using the rank-width~\cite{Hlineny-Oum-08,Oum-Seymour-06}.

Every graph of clique-width at most $k$ admits an {\em irredundant} $k$-expression, that is, a $k$-expression such that before any operation of the form $\eta_{a,b}$ is applied, the graph contains no edges between vertices
with label $a$ and vertices with label $b$~\cite{CO00}. In particular, this means that every operation $\eta_{a,b}$ adds at least one edge to the graph $G$.
Each expression $\sigma$ defines a rooted tree $T(\sigma)$, that we also call a {\it clique-width tree}.

\medskip
\noindent
{\bf Our result on graphs with bounded clique-width. }
We describe an algorithm for the {\sc $(\lambda, \beta, \alpha)$-TSS} problem
on graphs $G$ of clique-width at most $k$ given by an irredundant $k$-expression $\sigma$. Denoting by $n$
the number of vertices
of the input graph $G$, the running time of the algorithm is  bounded by
{$O(\lambda k|\sigma|{(n+1)^{(3\lambda +2)k}})$},
where $|\sigma|$ denotes the
{encoding length} of $\sigma$. For fixed $k$ and $\lambda$, this is polynomial in the size of the input.
We will first solve the following {\it decision problem} naturally associated with the
{\sc $(\lambda, \beta, \alpha)$-Target Set Selection} problem:

\medskip
\begin{sloppypar}
\noindent {\sc $(\lambda, \beta, \alpha)$-Target Set Decision ($(\lambda, \beta, \alpha)$-TSD)}.\\
{\bf Instance:} A graph $G=(V,E)$, thresholds $t:V\longrightarrow \mathbb{N}$, a latency bound $\lambda\in \N$,
a budget $\beta \in \N$ and an activation requirement $\alpha\in \N$.\\
{\bf Problem:} Determine whether there exists a set $S\subseteq V$ such
that $|S|\le \beta$ and \hbox{$|\Active[S,\lambda]|\ge \alpha$}.
\end{sloppypar}

Subsequently, we will argue how to modify the algorithm in order to solve the {\sc $(\lambda, \beta, \alpha)$-} and  the 
{\sc $(\lambda, \beta, A)$-Target Set Selection} problems.

Consider an instance $(G,t,\lambda, \beta, \alpha)$ to the {\sc $(\lambda, \beta, \alpha)$-Target Set Decision} problem,
where $G=(V,E)$ is a graph of clique-width at most $k$ given by an irredundant $k$-expression~$\sigma$.
We will develop a dynamic programming algorithm that will traverse the clique-width tree bottom up and simulate the activation process for the corresponding induced subgraphs of $G$, keeping track only of the minimal necessary information, that is, of how many vertices of each label become active in each round. For a bounded number of rounds $\lambda$, it will be possible to store and analyze the information in polynomial time. In order to compute these values recursively with respect to all the operations in the definition of the clique-width--including operations of the form $\eta_{a,b}$--we need to consider not only the original thresholds, but also reduced ones. This is formalized in Definition \ref{def1} below. We view $G$ as a $k$-labeled graph defined by $\sigma$. Given a $k$-labeled graph $H$ and a label $\ell\in [k]$, we denote by $V_\ell(H)$ the set of vertices of $H$ with label $\ell$.

\remove{
Subsequently, we will argue how to modify the algorithm in order to solve the {\sc $(\lambda, \beta, \alpha)$-} and  the 
{\sc $(\lambda, \beta, A)$-Target Set Selection} problems.
Consider an instance $(G,t,\lambda, \beta, \alpha)$ to the {\sc $(\lambda, \beta, \alpha)$-Target Set Decision} problem, where $G=(V,E)$
is a graph of clique-width at most $k$ given by an irredundant $k$-expression~$\sigma$.
In order to describe the algorithm, we need to introduce some notations and definitions.
We view $G$ as a $k$-labeled graph defined by $\sigma$. Given a $k$-labeled graph $H$ and a label $\ell\in [k]$, we denote by
$V_\ell(H)$ the set of vertices of $H$ with label $\ell$.
 The algorithm will rely on the following notion:
}

\begin{sloppypar}
\begin{definition} \label{def1}
Given a $k$-labeled subgraph $H$ of $G$ and a pair of matrices with non-negative integer entries
$({\boldsymbol\alpha},{\boldsymbol r})$ such that
${\boldsymbol\alpha} \in (\mathbb{Z}_+)^{[0,\lambda]\times [k]}$
(where $[0,\lambda]:= \{0,1,\ldots,\lambda\}$) and
${\boldsymbol r} \in (\mathbb{Z}_+)^{[\lambda]\times [k]}$,
an {\em $({\boldsymbol\alpha},{\boldsymbol r})$-activation process for $H$} is a non-decreasing
 sequence of vertex subsets $S[0]  \subseteq \ldots \subseteq S[\lambda]\subseteq V(H)$
such that the following conditions hold:  
\begin{enumerate}[(1)]
\item For every round $i \in [\lambda]$ and for every label $\ell\in [k]$,
  the set of all vertices with label $\ell$ activated at round $i$ is obtained with respect to the activation process
  starting at $S[0]$ with thresholds $t(u)$ reduced by $r[i,\ell]$ for all vertices with label $\ell$.
  Formally, for all $\ell\in [k]$
  and
  all $i \in [\lambda]$,
    $$(S[i] \setminus  S[i-1])\cap V_\ell(H) = \Big\{u \in V_\ell(H)\setminus  S[i-1]\,:\, \big|N_H(u)\cap S[i - 1]\big|\ge t(u)-r[i,\ell]\Big\}\,.$$
  \item For every label  $\ell\in [k]$, there are exactly $\alpha[0,\ell]$ initially activated vertices with label $\ell$: $|S[0]\cap  V_\ell(H)| = \alpha[0,\ell]\,.$
  \item For every label $\ell\in [k]$
  and for every round $i \in [\lambda]$, there are exactly $\alpha[i,\ell]$ vertices with label $\ell$ activated at round $i$: $\left|\left(S[i] \setminus  S[i-1]\right)\cap V_\ell(H)\right| = \alpha[i,\ell]\,.$
  \end{enumerate} 
\end{definition}
\end{sloppypar}

\remove{
The above definition will be essential for our dynamic programming algorithm, which will traverse
the clique-width tree bottom up and simulate the activation process for the corresponding induced subgraphs of $G$,
keeping track only of the minimal necessary information, that is, of how many vertices of each label become active in each round.
Operations of the form $\eta_{a,b}$ are the ones that require us to consider not only the original thresholds, but also the reduced ones.
For a bounded number of rounds $\lambda$, it will be possible to store and analyze the information
in polynomial time.
}

\begin{sloppypar}
Let ${\cal A}$ denote the set of all matrices of the form
${\boldsymbol\alpha} = \left(\alpha[i,\ell]\,:0\le i\le \lambda\,, 1\le \ell \le k\right)$
where $\alpha[i,\ell]\in [0,\alpha]$ for all $0\le i\le \lambda$ and all $1\le \ell \le k$.
Notice that $|{\cal A}| = (\alpha+1)^{(\lambda+1)k} = {O((n+1)^{(\lambda+1)k})}$.
Similarly, let ${\cal R}$ denote the set of all
matrices of the form
${\boldsymbol r} = \left(r[i,\ell]\,:1\le i\le \lambda\,, 1\le \ell \le k\right),$
where
$r[i,\ell]\in [0,n]$ for all $1\le i\le \lambda$ and all $1\le \ell\le k$.
Then \hbox{$|{\cal R}| = {(n+1)^{\lambda k}}$}.
\end{sloppypar}

Every node of the clique-width tree $T := T(\sigma)$ of the input graph $G$ corresponds to a $k$-labeled subgraph $H$ of $G$.
To every node of $T$ (and the corresponding $k$-labeled subgraph $H$ of $G$), we associate a Boolean-valued function
$\gamma_H:{\cal A}\times {\cal R}\longrightarrow\{0,1\}$
where
$\gamma_H({\boldsymbol\alpha},{\boldsymbol r}) = 1$
if and only if there exists an {$({\boldsymbol\alpha},{\boldsymbol r})$-activation process for $H$.
Each matrix pair $({\boldsymbol\alpha},{\boldsymbol r})\in {\cal A}\times {\cal R}$
can be described with $O(\lambda k)$ numbers.
Hence, the function $\gamma _H$ can be represented by storing the set of all
triples $\{({\boldsymbol\alpha},{\boldsymbol r},\gamma_H({\boldsymbol\alpha},{\boldsymbol r}))\,:\,({\boldsymbol\alpha},{\boldsymbol r})\in {\cal A}\times {\cal R}\}\,,$
requiring, in total, space
$$O(\lambda k)\cdot |{\cal A}\times {\cal R}| =
O(\lambda k)\cdot {O((n+1)^{(\lambda+1)k})\cdot O((n+1)^{\lambda k})} = {O(\lambda k(n+1)^{(2\lambda+1)k})}.$$

\medskip

Below we will describe how to compute all functions $\gamma_H$ for all subgraphs $H$ corresponding to the nodes of the
tree $T$. Assuming all these functions have been computed, we can extract the solution to the {\sc $(\lambda, \beta, \alpha)$-Target Set Decision} problem on $G$ from the root of $T$ as follows.
\begin{prop}\label{prop2}
There exists a set $S\subseteq V(G)$ such
that $|S|\le \beta$ and $|\Active[S,\lambda]|\ge \alpha$ if and only if
there exists a matrix $\boldsymbol\alpha\in {\cal A}$
with
$\gamma_G(\boldsymbol\alpha, \boldsymbol 0) = 1$
(where $\boldsymbol 0\in {\cal R}$ denotes the all zero matrix) such that
$\sum_{\ell = 1}^k\alpha[0,\ell] \le \beta$
and $\sum_{i = 0}^\lambda\sum_{\ell = 1}^k\alpha[i,\ell] \ge \alpha$.
\end{prop}

\begin{proof}{}
The constraint  $\sum_{\ell = 1}^k\alpha[0,\ell] \le \beta$ specifies that the total number of initially targeted vertices is within the budget $\beta$, and the constraint $\sum_{i = 0}^\lambda\sum_{\ell = 1}^k\alpha[i,\ell] \ge \alpha$ specifies that the total number of vertices activated within round $\lambda$
is at least the activation requirement $\alpha$.
\end{proof}

Here we give a detailed description of how to compute the functions $\gamma_H$ by traversing the tree $T$ bottom up.
We consider four cases according to the type of a node $v$ of the clique-width tree $T$.

{\bf Case 1: $v$ is a leaf.}\\
\hspace*{0.4truecm} In this case, the labeled subgraph $H$ of $G$ associated to $v$ is of the form $H = a(u)$ for some vertex $u\in V(G)$ and some label $a\in [k]$.
That is, a new vertex $u$ is introduced with label $a$.

Suppose that $({\boldsymbol\alpha},{\boldsymbol r})\in {\cal A}\times {\cal R}$ is a matrix pair such that
there exists an $({\boldsymbol\alpha},{\boldsymbol r})$-activation process ${\cal S} = (S[0],S[1],\ldots, S[\lambda])$ for $H$.
For every $\ell\in [k]\setminus \{a\}$, we have $V_\ell(H) = \emptyset$ and hence
$\alpha[i,\ell] = 0$ for all $i\in [0,\lambda]$.
Moreover, since $V_a(H) = \{u\}$, we have
$$0\le \sum_{i = 0}^\lambda\alpha[i,a]
= |S[0]\cap V_a(H)|+ \sum_{i = 1}^\lambda|(S[i]\setminus S[i-1])\cap V_a(H)|\le |V_a(H)| = 1\,.$$
Suppose first that $\sum_{i = 0}^\lambda\alpha[i,a] = 0$, that is, $\alpha[i,a] = 0$ for all $i$.
Then, $S[i] = \emptyset$ for all $i\in [0,\lambda]$,
and the defining property $(1)$ of the $({\boldsymbol\alpha},{\boldsymbol r})$-activation process implies
that $r[i,a]<t(u)$ for every $i\in [\lambda]$
(otherwise $u$ would belong to $S[i]$).

Now, suppose that $\sum_{i = 0}^\lambda\alpha[i,a] = 1$. Then, there exists a unique
$i^*\in [0,\lambda]$ such that $$\alpha[i,a] = \left\{
                                                            \begin{array}{ll}
                                                              1, & \hbox{if $i = i^*$;} \\
                                                              0, & \hbox{otherwise.}
                                                            \end{array}
                                                          \right.$$
If $i^* = 0$ then $\{u\} = S[0] \subseteq S[1] \subseteq \ldots \subseteq S[\lambda]\subseteq V(H) = \{u\}$, therefore
 $S[i] = \{u\}$ for all $i\in [0,\lambda]$, independently of $\boldsymbol r$.
If $i^* \ge 1$ then properties $(2)$ and $(3)$ imply that $S[0] = \ldots = S[i^*-1] = \emptyset$ and
$S[i^*] = S[i^*+1] = \ldots = S[\lambda] = \{u\}$. Hence, the
defining property $(1)$ of the $({\boldsymbol\alpha},{\boldsymbol r})$-activation process implies,
on the one hand, that $r[i,a]<t(u)$ for every $i\in \{1,\ldots, i^*-1\}$ (otherwise $u$ would belong to $S[i]$), while,
on the other hand, $r[i^*,a]\ge t(u)$. Hence, $i^* = \min\{i\ge 1\,:\, r[i,a]\ge t(u)\}$.

%
%
%
Hence, if there exists an $({\boldsymbol\alpha},{\boldsymbol r})$-activation process for $H$,
then
$({\boldsymbol\alpha},{\boldsymbol r})\in ({\cal A}\times {\cal R})^*$ where
\begin{eqnarray*}
({\cal A}\times {\cal R})^* &=& \bigg\{({\boldsymbol\alpha},{\boldsymbol r})\in {\cal A}\times {\cal R}\,:\,
(\forall \ell\neq a)(\alpha[i,\ell] = 0)\wedge \bigg(\sum_{i = 0}^\lambda\alpha[i,a] \le 1\bigg)\\
&&\wedge \bigg[\bigg(\sum_{i = 0}^\lambda\alpha[i,a] = 0\bigg) \Rightarrow \Big((\forall i)\big(r[i,a]<t(u)\big)\Big)\bigg]\\
&&\wedge \bigg[\Big((\exists i^*)(\alpha[i^*,a] = 1)\Big) \Rightarrow \Big(i^* = 0 \vee i^* = \min\{i\ge 1\,:\, r[i,a]\ge t(u)\}\Big)\bigg]\bigg\}\,.
\end{eqnarray*}

Conversely, by reversing the above arguments, one can verify that for every $({\boldsymbol\alpha},{\boldsymbol r})\in ({\cal A}\times {\cal R})^*$
there exists an $({\boldsymbol\alpha},{\boldsymbol r})$-activation process for $H$.
Hence, for every $({\boldsymbol\alpha},{\boldsymbol r})\in {\cal A}\times {\cal R}$, we set
$$\gamma_H({\boldsymbol\alpha},{\boldsymbol r}) = \left\{
                                                    \begin{array}{ll}
                                                      1, & \hbox{if $({\boldsymbol\alpha},{\boldsymbol r})\in ({\cal A}\times {\cal R})^*$;} \\
                                                      0, & \hbox{otherwise.}
                                                    \end{array}
                                                  \right.$$

\medskip
{\bf Case 2: $v$ has exactly two children in $T$.}\\
\hspace*{0.4truecm}In this case, the labeled subgraph $H$ of $G$ associated to $v$ is the disjoint union $H = H_1\oplus H_2$, where $H_1$ and $H_2$ are the labeled subgraphs of
$G$ associated to the two children of $v$ in $T$.

Suppose that
$(S[0],\ldots, S[\lambda])$ is an $({\boldsymbol \alpha}, {\boldsymbol r})$-activation process for $H$.
For every round $i\in [0,\lambda]$ and for every label $\ell\in [k]$, set
$$S_1[i] = S[i]\cap V(H_1)\,,$$ and
$${ \alpha}_1[i,\ell] = \left\{
                                    \begin{array}{ll}
                                      |S_1[0]\cap V_\ell(H_1)|, & \hbox{if $i = 0$;} \\
                                      |(S_1[i]\setminus S_1[i-1])\cap V_\ell(H_1)|, & \hbox{otherwise.}
                                    \end{array}
      \right.$$
Then,
$(S_1[0],\ldots, S_1[\lambda])$ is an $({\boldsymbol \alpha}_1, {\boldsymbol r})$-activation process for $H_1$.
Properties $(2)$ and $(3)$ follow immediately from the definition of ${\boldsymbol \alpha}_1$.
Property $(1)$ follows from the fact that in $H$ there are no edges between vertices of $H_1$ and $H_2$.
One can analogously define an $({\boldsymbol \alpha}_2, {\boldsymbol r})$-activation process
for $H_2$.
Since $H$ is the disjoint union of $H_1$ and $H_2$,
these two processes satisfy the matrix equation  ${\boldsymbol \alpha}_1+{\boldsymbol \alpha}_2 = {\boldsymbol \alpha}$.

Conversely,
suppose that there exist an  $({\boldsymbol \alpha}_1, {\boldsymbol r})$-activation process
$(S_1[0],\ldots, S_1[\lambda])$ for $H_1$ and
an  $({\boldsymbol \alpha}_2, {\boldsymbol r})$-activation process
$(S_2[0],\ldots, S_2[\lambda])$ for $H_2$.
Then, defining $S[i] = S_1[i]\cup S_2[i]$ for all rounds $i\in [0,\lambda]$, we obtain
an $({\boldsymbol \alpha}, {\boldsymbol r})$-activation process
$(S[0],\ldots, S[\lambda])$ for $H$,
where
${\boldsymbol \alpha}={\boldsymbol \alpha}_1+{\boldsymbol \alpha}_2$.

Hence, for every
$({\boldsymbol\alpha},{\boldsymbol r})\in {\cal A}\times {\cal R}$ we set
$$\gamma_H({\boldsymbol\alpha},{\boldsymbol r}) =
\left\{
  \begin{array}{ll}
    1, & \hbox{if $(\exists {\boldsymbol \alpha}_1, {\boldsymbol \alpha}_2\in {\cal A})({\boldsymbol \alpha}={\boldsymbol \alpha}_1+{\boldsymbol \alpha}_2$ and
$\gamma_{H_1}({\boldsymbol\alpha}_1,{\boldsymbol r})=\gamma_{H_2}({\boldsymbol\alpha}_2,{\boldsymbol r}) = 1$);} \\
    0, & \hbox{otherwise.}
  \end{array}
\right.$$

{\bf Case 3: $v$ has exactly one child in $T$ and the labeled subgraph $H$ of $G$ associated to $v$ is of the form $H = \eta_{a,b}(H_1)$.}\\
\hspace*{0.4truecm}In this case,  graph $H$ is obtained from $H_1$ by adding all edges between vertices labeled $a$ and vertices labeled $b$.
Since the $k$-expression is irredundant, in $H_1$ there are no edges between vertices labeled $a$ and vertices labeled $b$.

Suppose that
${\cal S} = (S[0],\ldots, S[\lambda])$ is an $({\boldsymbol \alpha}, {\boldsymbol r})$-activation process for $H$.
For every round $i\in [0,\lambda]$ and for every label $\ell\in [k]$, set
\begin{equation*}
{r}_1[i,\ell] = \left\{
                                    \begin{array}{ll}
                                      {\min\{n, r[i,a]+\sum_{j < i}\alpha[j,b]\}}, & \hbox{if $\ell= a$;} \\
                                      {\min\{n, r[i,b]+\sum_{j < i}\alpha[j,a]\}}, & \hbox{if $\ell= b$;} \\
                                      r[i,\ell], & \hbox{otherwise\,,}
                                    \end{array}
      \right.
\end{equation*}
Let us verify that ${\cal S}$ is an $({\boldsymbol \alpha}, {\boldsymbol r}_1)$-activation process for $H_1$:
\begin{itemize}
  \item Defining conditions $(2)$ and $(3)$ are satisfied since the partition of the vertex set $V(H) = V(H_1)$ into label classes
is the same in both graphs $H$ and $H_1$.
  \item To verify condition $(1)$, notice first that for every label $\ell\in [k]\setminus\{a,b\}$ and every vertex $u\in V_\ell(H_1) = V_\ell(H)$, we have $N_{H_1}(u) = N_H(u)$. Moreover, for each round $i\in [\lambda]$, it holds that $r_i[i,\ell] = r[i,\ell]$, which implies
   $\big|N_{H_1}(u)\cap S[i - 1]\big|\ge t(u)-r_1[i,\ell]$ if and only if
   $\big|N_{H}(u)\cap S[i - 1]\big|\ge t(u)-r[i,\ell]$.

Now consider the case $\ell = a$. (The case $\ell = b$ is analogous.)
Since the $k$-expression is irredundant,
   the $H$-neighborhood of every vertex $u\in V_a(H_1) = V_a(H)$ is equal to the disjoint union
   $$N_H(u) = N_{H_1}(u) \cup V_b(H)\,.$$
Consider an arbitrary round $i\in [\lambda]$.
We will show that condition
\begin{equation}\label{cond1}
\big|N_{H_1}(u)\cap S[i - 1]\big|\ge t(u)-r_1[i,a]
\end{equation}
is equivalent to the condition
\begin{equation}\label{cond2}
\big|N_{H}(u)\cap S[i - 1]\big|\ge t(u)-r[i,a]\,.
\end{equation}
The set $S[i - 1]$ can be written as the disjoint union
$$S[i - 1] = S[0]\cup\bigcup_{j = 1}^{i-1}\left(S[j]\setminus S[j-1]\right)\,,$$
hence
\begin{eqnarray*}
\big|S[i - 1]\cap V_{b}(H)\big| &=& \big|S[0]\cap V_{b}(H)\big|+\sum_{j = 1}^{i-1}\big|(S[j]\setminus S[j-1])\cap V_{b}(H)\big|
\\
&=& \alpha[0,b]+\sum_{j = 1}^{i-1}\alpha[j,b] = \sum_{j<i}\alpha[j,b]
\end{eqnarray*}
and consequently
\begin{eqnarray*}
\big|N_{H}(u)\cap S[i - 1]\big| &=& \big|N_{H_1}(u)\cap S[i - 1]\big| + \big|V_{b}(H)\cap S[i - 1]\big|\\
&=& \big|N_{H_1}(u)\cap S[i - 1]\big| + \sum_{j<i}\alpha[j,b]\,.
\end{eqnarray*}


{Suppose first that $t(u)\le r_1[i,a]$. Then,
condition (\ref{cond1}) trivially holds, and condition (\ref{cond2}) holds as well:
\begin{eqnarray*}
\big|N_{H}(u)\cap S[i - 1]&\ge& \big|S[i - 1]\cap V_{b}(H)\big|= \sum_{j<i}\alpha[j,b]\\
&=&(r[i,a]+\sum_{j<i}\alpha[j,b])-r[i,a]\\
&\ge& r_1[i,a]-r[i,a]\ge t(u)-r[i,a]\,.
\end{eqnarray*}
}
{Suppose now that $t(u) > r_1[i,a]$. Then, we have  $r_1[i,a]<n$, which implies that
$r_1[i,a] = r[i,a]+\sum_{j<i}\alpha[j,b]$.}
Therefore, condition (\ref{cond1}),
$$\big|N_{H_1}(u)\cap S[i - 1]\big|\ge t(u)-r_1[i,a]\,,$$
is equivalent to the condition
$$\big|N_{H_1}(u)\cap S[i - 1]\big|\ge t(u)-r[i,a]-\sum_{j<i}\alpha[j,b]$$
which is in turn equivalent to
$$\big|N_{H}(u)\cap S[i - 1]\big|\ge t(u)-\bigg(r[i,a]+\sum_{j<i}\alpha[j,b]\bigg)+ \sum_{j<i}\alpha[j,b]$$
which is the same as condition (\ref{cond2})\,,
$$\big|N_{H}(u)\cap S[i - 1]\big|\ge t(u)-r[i,a]\,.$$

Putting the two cases together, we have
\begin{eqnarray*}
&&(S[i] \setminus  S[i-1])\cap V_a(H_1)= (S[i] \setminus  S[i-1])\cap V_a(H)\\
& =& \Big\{u \in V_a(H)\setminus  S[i-1]\,:\, \big|N_{H}(u)\cap S[i - 1]\big|\ge t(u)-r[i,a]\Big\}\\
& =& \Big\{u \in V_a(H_1)\setminus  S[i-1]\,:\, \big|N_{H_1}(u)\cap S[i - 1]\big|\ge t(u)-r_1[i,a]\Big\}\,,
\end{eqnarray*}
and ${\cal S}$ is indeed an $({\boldsymbol \alpha}, {\boldsymbol r}_1)$-activation process for $H_1$.
\end{itemize}

Conversely, suppose that $({\boldsymbol\alpha},{\boldsymbol r})\in {\cal A}\times {\cal R}$ is such that
${\cal S} = (S[0],\ldots, S[\lambda])$ is an $({\boldsymbol \alpha}, {\boldsymbol r}_1)$-activation process for $H_1$, where
\begin{equation*}
{r}_1[i,\ell] = \left\{
                                    \begin{array}{ll}
                                      {\min\{n, r[i,a]+\sum_{j < i}\alpha[j,b]\}}, & \hbox{if $\ell= a$;} \\
                                      {\min\{n, r[i,b]+\sum_{j < i}\alpha[j,a]\}}, & \hbox{if $\ell= b$;} \\
                                      r[i,\ell], & \hbox{otherwise\,,}
                                    \end{array}
      \right.
\end{equation*}
Reversing the argument above shows that ${\cal S}$ is an $({\boldsymbol \alpha}, {\boldsymbol r})$-activation process for~$H$.

Hence, for every $({\boldsymbol\alpha},{\boldsymbol r})\in {\cal A}\times {\cal R}$ we define
the integer-valued matrix ${\boldsymbol r}_1$
by setting
\begin{equation*}
{r}_1[i,\ell] = \left\{
                                    \begin{array}{ll}
                                      {\min\{n, r[i,a]+\sum_{j < i}\alpha[j,b]\}}, & \hbox{if $\ell= a$;} \\
                                      {\min\{n, r[i,b]+\sum_{j < i}\alpha[j,a]\}}, & \hbox{if $\ell= b$;} \\
                                      r[i,\ell], & \hbox{otherwise\,,}
                                    \end{array}
      \right.
\end{equation*}
for every round $i\in [0,\lambda]$ and for every label $\ell\in [k]$.
Then, we set, for all $({\boldsymbol\alpha},{\boldsymbol r})\in {\cal A}\times {\cal R}$,
$$\gamma_H({\boldsymbol\alpha},{\boldsymbol r}) = \gamma_{H_1}({\boldsymbol\alpha},{\boldsymbol r}_1)\,.$$

{\bf Case 4: $v$ has exactly one child in $T$ and
the labeled subgraph $H$ of $G$ associated to $v$ is of the form $H = \rho_{a\to b}(H_1)$.}

Suppose that $({\boldsymbol \alpha}, {\boldsymbol r})\in {\cal A}\times {\cal R}$ is such that
there exists an $({\boldsymbol \alpha}, {\boldsymbol r})$-activation process ${\cal S} = (S[0],\ldots, S[\lambda])$ for $H$.
For every round $i\in [0,\lambda]$ and for every label $\ell\in [k]$, set
$${ \alpha}_1[i,\ell] = \left\{
                                    \begin{array}{ll}
                                      |S[0]\cap V_\ell(H_1)|, & \hbox{if $i = 0$;} \\
                                      |(S[i]\setminus S[i-1])\cap V_\ell(H_1)|, & \hbox{otherwise.}
                                    \end{array}
      \right.$$
and, for every round $i\in [\lambda]$ and for every label $\ell\in [k]$, set
$${r}_1[i,\ell] = \left\{
                                    \begin{array}{ll}
                                      r[i,b], & \hbox{if $\ell = a$;} \\
                                      r[i,\ell], & \hbox{otherwise.}
                                    \end{array}
      \right.$$
Then, ${\cal S}$ is an $({\boldsymbol \alpha}_1, {\boldsymbol r}_1)$-activation process for $H_1$:
Properties $(2)$ and $(3)$ follow immediately from the definition of ${\boldsymbol \alpha}_1$.
To verify property $(1)$, let $i \in [\lambda]$ and $\ell\in [k]$.
If $\ell \not\in \{a,b\}$ then $V_\ell(H_1) = V_\ell(H)$ and $r_1[i,\ell] = r[i,\ell]$, hence the condition in
property $(3)$ holds in this case.
If $\ell \in \{a,b\}$ then, since $V_\ell(H_1)\subseteq V_b(H)$, we have
\begin{eqnarray*}
   && \big(S[i] \setminus  S[i-1]\big)\cap V_\ell(H_1)\\
   &=& \Big(\big(S[i] \setminus  S[i-1]\big)\cap V_b(H)\Big)\cap V_\ell(H_1) \\
   &=&  \Big\{u \in V_b(H)\setminus  S[i-1]\,:\, \big|N_{H}(u)\cap S[i - 1]\big|\ge t(u)-r[i,b]\Big\}\cap V_\ell(H_1)\\
   &=&  \Big\{u \in V_\ell(H_1)\setminus  S[i-1]\,:\, \big|N_{H_1}(u)\cap S[i - 1]\big|\ge t(u)-r_1[i,\ell]\Big\}\,,
\end{eqnarray*}
so again the condition holds.
Notice that the matrices $\boldsymbol \alpha$ and $\boldsymbol \alpha_1$ are related as follows:
For every round $i\in [0,\lambda]$ and for every label $\ell\in [k]$, we have
$${\alpha}[i,\ell] = \left\{
                        \begin{array}{ll}
                          0, & \hbox{if $\ell = a$;} \\
                          {\alpha}_1[i,a]+{\alpha}_1[i,b], & \hbox{if $\ell = b$;}\\
                          {\alpha}_1[i,\ell], & \hbox{otherwise}.
                          \end{array}
                      \right.$$

Conversely, suppose that $({\boldsymbol \alpha}, {\boldsymbol r})\in {\cal A}\times {\cal R}$ is such that
there exists an  $({\boldsymbol \alpha}_1, {\boldsymbol r}_1)$-activation process
${\cal S} = (S[0],\ldots, S[\lambda])$ for $H_1$, where
for every round $i\in [\lambda]$ and for every label $\ell\in [k]$, we have
$${r}_1[i,\ell] = \left\{
                                    \begin{array}{ll}
                                      r[i,b], & \hbox{if $\ell = a$;} \\
                                      r[i,\ell], & \hbox{otherwise.}
                                    \end{array}
      \right.$$
and
for every $i\in [0,\lambda]$ and for every label $\ell\in [k]$, we have
$${\alpha}[i,\ell] = \left\{
                        \begin{array}{ll}
                          0, & \hbox{if $\ell = a$;} \\
                          {\alpha}_1[i,a]+{\alpha}_1[i,b], & \hbox{if $\ell = b$;}\\
                          {\alpha}_1[i,\ell], & \hbox{otherwise}.
                          \end{array}
                      \right.$$
Then, it can be verified that ${\cal S}$ is an $({\boldsymbol \alpha}, {\boldsymbol r})$-activation process for $H$.

Hence, for every
$({\boldsymbol\alpha},{\boldsymbol r})\in {\cal A}\times {\cal R}$ we set
$\gamma_H({\boldsymbol\alpha},{\boldsymbol r})=1$ if and only if
there exists $({\boldsymbol\alpha}_1,{\boldsymbol r}_1)\in {\cal A}\times {\cal R}$ such that
$\gamma_{H_1}({\boldsymbol\alpha}_1,{\boldsymbol r}_1) = 1$, where
for every round $i\in [\lambda]$ and for every label $\ell\in [k]$, we have
$${r}_1[i,\ell] = \left\{
                                    \begin{array}{ll}
                                      r[i,b], & \hbox{if $\ell = a$;} \\
                                      r[i,\ell], & \hbox{otherwise.}
                                    \end{array}
      \right.$$
and
for every $i\in [0,\lambda]$ and for every label $\ell\in [k]$, we have
$${\alpha}[i,\ell] = \left\{
                        \begin{array}{ll}
                          0, & \hbox{if $\ell = a$;} \\
                          {\alpha}_1[i,a]+{\alpha}_1[i,b], & \hbox{if $\ell = b$;}\\
                          {\alpha}_1[i,\ell], & \hbox{otherwise}.
                          \end{array}
                      \right.$$
This completes the description of the four cases and with it the description of the algorithm.

\noindent {\bf Correctness and time complexity.} Correctness of the algorithm follows from the derivation of the recursive formulas.
We now analyze the algorithm's time complexity. Given an irredundant $k$-expression $\sigma$ of $G$, the clique-width tree $T$ can be computed from $\sigma$ in linear time.
The algorithm computes the sets ${\cal A}$ and ${\cal R}$ in time
$|{\cal A}| = {O((n+1)^{(\lambda+1)k})}$
and
$|{\cal R}| = {O((n+1)^{\lambda k})}\,,$
respectively.

The algorithm then traverses the clique-width tree bottom-up.
At each leaf of $T$ and for each $({\boldsymbol \alpha}, {\boldsymbol r})\in {\cal A}\times {\cal R}$,
it can be verified in time $O(\lambda k)$ whether $({\boldsymbol \alpha}, {\boldsymbol r})\in ({\cal A}\times {\cal R})^*$.
Hence, the function $\gamma_H$ at each leaf can be computed in time
${O(\lambda k(n+1)^{(2\lambda +1)k})}$.

\begin{sloppypar}
At an internal node corresponding to Case 2, the value of
$\gamma_H({\boldsymbol \alpha}, {\boldsymbol r})$ for a given $({\boldsymbol \alpha}, {\boldsymbol r})\in {\cal A}\times {\cal R}$
can be computed in time $O(|{\cal A}|\lambda k)$ by iterating over all ${\boldsymbol \alpha}_1\in {\cal A}$,
verifying whether ${\boldsymbol \alpha}_2 := {\boldsymbol \alpha}-{\boldsymbol \alpha}_1\in {\cal A}$
and looking up the values of $\gamma_{H_1}({\boldsymbol \alpha}_1, {\boldsymbol r})$ and $\gamma_{H_2}({\boldsymbol \alpha}_2, {\boldsymbol r})$.
Hence, the total time spent at an internal node corresponding to Case 2 is
$$O(|{\cal A}|\lambda k)\cdot {O((n+1)^{(2\lambda +1)k}) = O(\lambda k(n+1)^{(3\lambda +2)k})}.$$
\end{sloppypar}

At an internal node corresponding to Case 3 or Case 4, the value of
$\gamma_H({\boldsymbol \alpha}, {\boldsymbol r})$ for a given $({\boldsymbol \alpha}, {\boldsymbol r})\in {\cal A}\times {\cal R}$
can be computed in time $O(\lambda k)$. Hence, the total time spent at any such node is
${O(\lambda k(n+1)^{(2\lambda +1)k})}$.

The overall time complexity is $O(\lambda k|\sigma|{(n+1)^{(3\lambda +2)k}})$. For fixed $k$ and $\lambda$, this is polynomial
in the size of the input.

Given the above algorithm for the {\sc $(\lambda, \beta, \alpha)$-Target Set Decision} problem on graphs of bounded clique-width,
finding a set $S$ that solves the
{\sc $(\lambda, \beta, \alpha)$-Target Set Selection} problem
can be done by standard backtracking techniques.
We only need to extend the above algorithm so that at every node node of the clique-width tree $T$
(and the corresponding $k$-labeled subgraph $H$ of $G$) and every
$({\boldsymbol\alpha},{\boldsymbol r})\in {\cal A}\times {\cal R}$
such that $\gamma_H({\boldsymbol\alpha},{\boldsymbol r}) = 1$,
the algorithm also keeps track of an $({\boldsymbol\alpha},{\boldsymbol r})$-activation process for $H$.
As shown in the above analysis of Cases 1--4, this can be computed in polynomial time
using the recursively computed $({\boldsymbol\alpha},{\boldsymbol r})$-activation processes.
Hence, we have the following theorem.

\begin{theorem}\label{thm1}
For every fixed $k$ and $\lambda$, the
{\sc $(\lambda, \beta, \alpha)$-Target Set Selection} problem
can be solved in polynomial time on graphs of clique-width at most $k$.
\end{theorem}

\noindent {When $\lambda = 1$ and $\alpha = |V(G)|$, the
{\sc $(\lambda, \beta, \alpha)$-Target Set Selection} problem coincides with the {\sc Vector Domination} problem
(see, e.g,~\cite{CMV12+}). Hence, Theorem~\ref{thm1} answers a question
from~\cite{CMV12+} regarding the complexity status of {\sc Vector Domination}
for graphs of bounded treewidth or bounded clique-width.}

\medskip
\noindent
{\bf The {\sc $(\lambda, \beta, A)$-TSS} problem on graphs of small clique-width. }
The approach to solve
the {\sc $(\lambda, \beta, A)$-Target Set Selection} problem on graphs of bounded clique-width
is similar to the one above.
First, we consider the decision problem} naturally associated with the
{\sc $(\lambda, \beta, A)$-TSS} problem, the
{\sc $(\lambda, \beta, A)$-Target Set Decision} problem ({\sc $(\lambda, \beta, A)$-TDS} for short).
Consider an instance $(G,t,\lambda, \beta, A)$ to the {\sc $(\lambda, \beta, A)$-TSD} problem, where $G=(V,E)$
is a graph of clique-width at most $k$ given by an irredundant $k$-expression~$\sigma$.
First, we construct a $2k$-expression $\sigma'$ in such a way that every labeled vertex $a(u)$ with $u\in A$ changes to $(a+k)(u)$.
Moreover, every operation of the form $\eta_{i,j}$ is replaced with a sequence of four composed operations
$\eta_{i,j}\circ \eta_{i,j+k}\circ \eta_{i+k,j}\circ \eta_{i+k,j+k}$, and every operation of the form
$\rho_{i\to j}$ is replaced with a sequence of two composed operations
$\rho_{i,j}\circ \rho_{i+k,j+k}$.
The so defined expression $\sigma'$ can be obtained from $\sigma$ in linear time, and
defines a labeled graph isomorphic to $G$ such that the set $A$ contains precisely the vertices with labels
strictly greater than $k$.
Using the same notation as above (with respect to $\sigma'$), we obtain the following
%
\begin{prop}
There exists a set $S\subseteq V(G)$ such that $|S|\le \beta$ and $|\Active[S,\lambda]|\supseteq A$ if and only if
there exists a matrix $\boldsymbol\alpha\in {\cal A}$ with
$\gamma_G(\boldsymbol\alpha, \boldsymbol 0) = 1$ such that
$\sum_{\ell = 1}^k\alpha[0,\ell] \le \beta$
and $\sum_{i = 0}^\lambda\sum_{\ell = k+1}^{2k}\alpha[i,\ell] = |A|$.
\end{prop}
Hence, the same approach as above can be used to solve first the
{\sc $(\lambda, \beta, A)$-Target Set Decision} problem,  and then  the
{\sc $(\lambda, \beta, A)$-Target Set Selection} problem itself.
\begin{theorem}\label{thm2}
For every fixed $k$ and $\lambda$, the {\sc $(\lambda, \beta, A)$-Target Set Selection} problem
can be solved in polynomial time on graphs of clique-width at most $k$.
\end{theorem}
\begin{remk}
The dependency on $\lambda$ and $k$ in Theorems~\ref{thm1} and~\ref{thm2} is exponential.
Since the Vector Dominating Set problem (a special case of {\sc $(\lambda, \beta, \alpha)$-Target Set Selection} problem)
is W[1]-hard with respect to the parameter treewidth~\cite{BBNU11},
the exponential dependency on $k$ is most likely unavoidable.
We leave open the question whether the {\sc $(\lambda, \beta, \alpha)$-} and {\sc $(\lambda, \beta, A)$-Target Set Selection} problems are FPT (or even polynomial) with respect to parameter $\lambda$ for graphs of bounded treewidth or clique-width.
\end{remk}

	
\section{$(\lambda,A)-$TSS on Trees}\label{sec:trees} 	
Since trees are graphs of clique-width at most~$3$, results of Section~\ref{sec:cwd} imply that
the {\sc $(\lambda, \beta, \alpha)$-} and
{\sc $(\lambda, \beta, A)$-TSS} problems
are solvable in polynomial time on trees when  $\lambda$ is constant.
In this section we improve on this latter result by giving a \emph{linear time} algorithm
for the {\sc $(\lambda,A)-$TSS problem}, for \emph{arbitrary} values of $\lambda$. {Our result also extends the}
linear time solution for the classical TSS problem (i.e., without the latency bound) on trees proposed in \cite{Chen-09}.
Like the solution in \cite{Chen-09},
we will assume that the tree is rooted at some node $r$. Then,
once such rooting is fixed, for any node $v$ we will denote by $T(v)$ the  subtree rooted at $v,$
by $C(v)$ the set of children of $v$ and, for $v \neq r,$ by $p(v)$ the parent of $v.$

In the following we assume that $\forall v \in V, 1 \leq t(v) \leq d(v)$.
The more general case (without these assumptions) can be handled with minor changes to the proposed algorithm.

The algorithm $(\lambda,A)-$\textbf{TSS on Trees} on p.~\pageref{algo} considers each node for being included in the target set $S$ in a bottom-up fashion. Each node is considered after all its children.
Leaves are never added to $S$ because there is always an optimal solution
in which the target set  consists of  internal nodes only. Indeed,
since all  leaves have thresholds equal to $1$, starting from any target  set
containing  some leaves we can get a solution of at most the same size by substituting each targeted leaf by its parent.

Thereafter, for each non-leaf node $v$, the algorithm checks whether the partial solution $S$ constructed so far allows to activate all the nodes in $T(v)\cap A$ (where $A$ is the set of nodes which must be activated) within round $\lambda$: the algorithm computes the round $\tau=\lambda-\tpth{v}$  by which $v$ has to be activated (line 12 of the pseudocode), where $\tpth{v}$ denotes the maximum length of a path from  $v$ to one of its descendants which  requires $v$'s influence to become active by round $\lambda$.  Notice that   $\tau< \lambda$ when
there exists a vertex in the subtree $T(v)$ which has to be activated
by time $\lambda$, and this can  happen only if $v$ is activated by time $\tau$. Then the algorithm computes the set $Act(v)$ consisting of those $v$'s children which are activated at round $\tau -1$ (line 13).
The algorithm is based on the following three observations \textbf{(a)}, \textbf{(b)}, and \textbf{(c)}
(assuming that $v$ is in the set of nodes which must be activated):

\textbf{(a)} $v$ must be included in the target set solution $S$  whenever the nodes belonging to  $Act(v)\cup \{p(v)\}$ do not suffice to activate $v,$ i.e., the current partial solution is such  that  at most $t(v)-2$ children of $v$ can be active at round $\tau -1.$

\textbf{(b)} $v$ must be included in $S$  if $\tau=0$  (i.e., $\lambda=\tpth{v}$). Indeed, in this case, there exists a vertex in $T(v)$, at distance $\lambda$ from $v$, which requires $v$'s influence to be activated, and this can only happen if $v$ is activated at round $0$.

These two cases for the activation of $v$ are taken care by lines 19-21 of the pseudocode. If neither \textbf{(a)} nor \textbf{(b)} is verified, then $v$ is not activated. However, it might be that the algorithm has to  guarantee the activation of some other node in the subtree $T(v)$. To deal with  such a case, when

\textbf{(c)} the size of the set $Act(v)$ is  $t(v)-1$,  then  the algorithm puts  $p(v)$ in the set $A$ of nodes to be activated; moreover,   the value of the parameter $\pth{v}$ is updated coherently in such a way to correctly compute the value of $\tpth{p(v)}$  which  assures that  $p(v)$ gets active within round $\lambda-\tpth{p(v)}$ (see lines 22-24).

For the root of the tree, which has no parent,  case \textbf{(c)}  is managed as  case \textbf{(a)} (see lines 26-28).

In order to keep track of the above cases while traversing the tree bottom-up, the algorithm uses the following parameters:

\noindent -- $\tm{v}$ assume value equal to the round (of the activation process with target set $S$) in which $v$ would be activated only
thanks to its children and irrespectively of the status of its parent. Namely,
$\tm{v}= \infty$ \ { if  $v$  is a leaf}, $\tm{v}=0$ if $v \in S$, and
$\tm{v}=1+ \min^{t(v)} \{\tm{u}  \ | \ u \in  C(v) \}$    { otherwise}.
Here  $\min^{t(v)} C$ denotes the $t(v)$--th smallest element in the set $C$.

\noindent -- $\pth{v}$  assume value equal to
$-1$ in case  $v$'s parent is not among the activators of $v$; otherwise,  assume value equal to
the maximum length of a path from  $v$ to one of its descendants which
(during  the activation process with target set $S$) requires $v$'s influence in order to become active.
It will be shown that    during the activation
process with target set $S$, for each node $v\in A$  we have $v \in  \Actt{S}{\min\{\lambda- \max_{u \in C(v)} \pth{u}-1, \tm{v}\}},$
for each node $v \in A$. Moreover, the algorithm  maintains a set $A'\supseteq A$ of nodes to be activated. Initially $A'=A$,
the set $A'$ can be enlarged  when the algorithm decides not to include in $S$ the node $v$ under consideration  but to use $p(v)$ for $v$'s activation, like in the case \textbf{(c)} above.

In the rest of the section, we prove Theorem~\ref{teo-alberi}.
\begin {theorem}\label{teo-alberi}
Algorithm \textbf{$(\lambda,A)-$TSS on Tree} computes, in time $O(|V|)$, an optimal solution for the {\sc $(\lambda,A)-$ Target Set Selection} problem  
on a tree.
\end{theorem}


		\begin{algorithm}\label{algo}
\SetAlgoInsideSkip{footnotesize}
 \SetCommentSty{footnotesize}
\SetAlFnt{footnotesize}
\SetKwInput{KwData}{Input}
\SetKwInput{KwResult}{Output}
\caption{ \ \   $(\lambda,A)-$\textbf{TSS on Trees}\label{alg:relation2}}
\KwData { A tree $T=(V,E)$,  thresholds function $t : V \rightarrow \mathbb{N}$, a latency bound $\lambda \in \mathbb{N}$ and  a set to be activated $A\subseteq V$.}
\KwResult{ $S\subseteq V$ of minimum size such that $A \subseteq Active[S,\lambda] .$}\BlankLine
$S= \emptyset;$\\
$A'=A;$\\
Fix a root $r\in V$ \tcp*[f]{ $T(r)$ denotes the tree $T$ rooted at $r$}\\
\ForAll { \em $v$ in the set of \ $T(r)$ leaves}{
$\tm{v}=\infty$\\
\eIf(\tcp*[f]{ $v$ belongs to the set of nodes to be activated}){$v\in A'$}{
$A'=A'\cup \{p(v)\}$ \tcp*[f]{$p(v)$ denotes $v$'s parent}\\
$\pth{v}=0$ 
}{$\pth{v}=-1$}
}
\ForAll { \em $v$ in the set of \ $T(r)$ internal nodes, listed in reverse order with respect to the time they are visited by a breadth-first traversal from $r$  }
{	$\tpth{v}=1+\max_{u \in C(v)}\pth{u} $ \tcp*[f]{$C(v)$ is the set of $v$'children}\\
		$\A{v}=\{ u \in C(v) \ | \ \tm{u} < \lambda-\tpth{v} \}$ \\
		 $\pth{v}=-1$\\
	   $\tm{v}=1+ \min^{t(v)} \{\tm{u}  \ | \ u \in  C(v) \}$ \\	
		{\If(\tcp*[f]{$v$ has to be activated}){$v \in A'$}{
			    \eIf{$v\neq r$} {
    		 \Switch{}{
		    	  \Case(){$(|\A{v}| {\leq} t(v){-}2)$ OR $(\tpth{v}=\lambda)$}
		    	  		{$S=S\cup \{v\}$ \tcp*[f]{$v$ has to be in the target set}  \\
    			$\tm{v}=0$}
	    	  \Case(){$(|\A{v}|=t(v){-}1)$  AND $(\tpth{v}<\lambda)$}
	    	  		{$A'=A'\cup \{p(v)\}$ \tcp*[f]{$v$ will be activated thanks to its parent $p(v)$}\\
	    	  		$\pth{v}=\tpth{v}$
	    	  		}
		    	  }
       }(\tcp*[f]{ $v$ is the root})		
  			{\If{$(|\A{v}|\leq t(v){-}1)$}{
         	$S=S\cup \{v\}$\\
         	$\tm{v}=0$
    		}
    		}
    	}
   }
   }
 \textbf{return} ($S$)
 \end{algorithm}

\begin{figure}
\centerline{\epsfig{figure=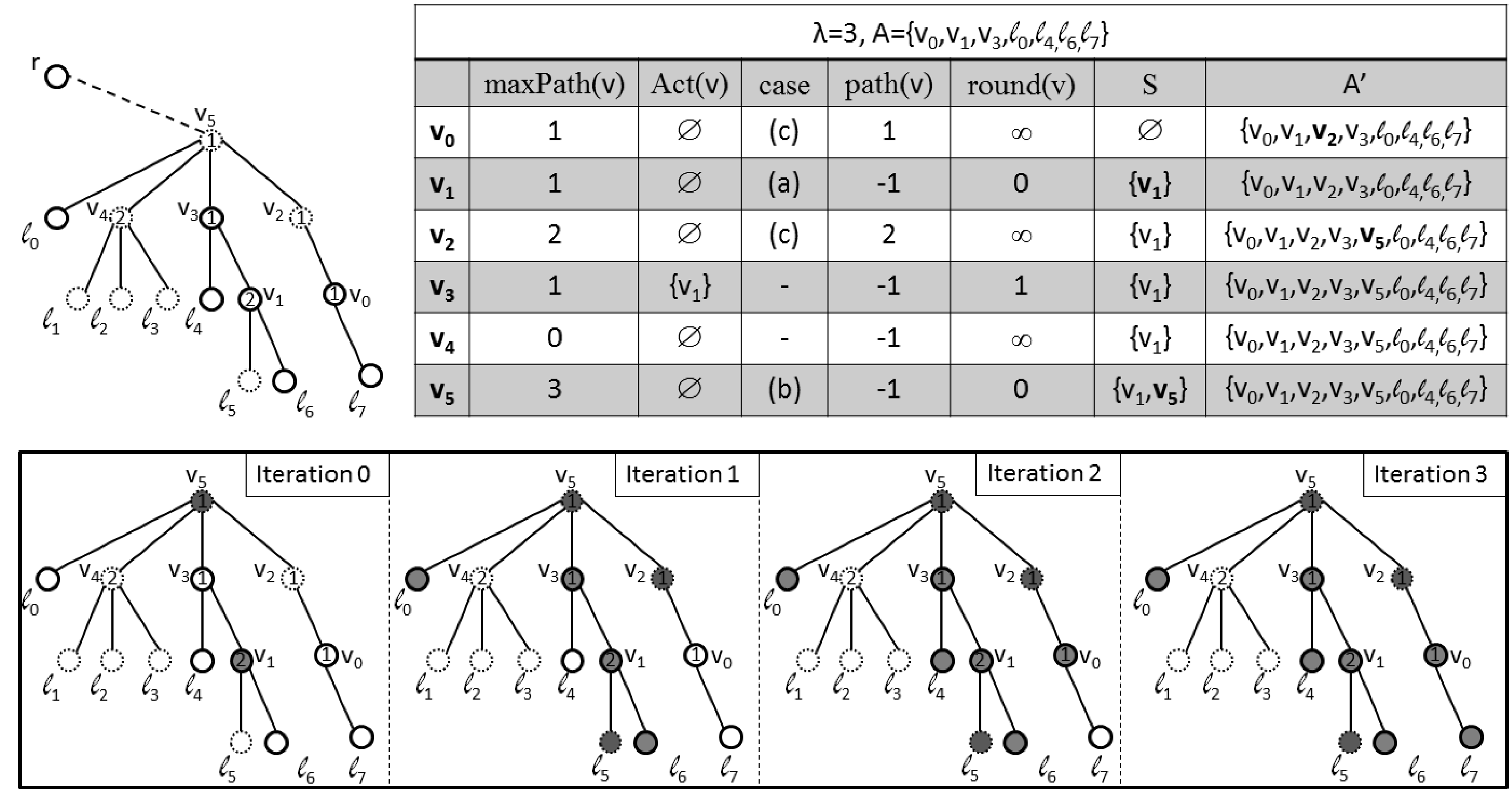,width=13.0truecm,height=5.7truecm}}
	\caption{  An example of execution of the algorithm $(\lambda,A)-$\textbf{TSS on Tree}:
 (left) a subtree rooted in $v_5$ (subscripts describe the order in which nodes are analyzed by the algorithm),  each node is depicted as a circle and  its threshold is given inside the circle. Circles having a solid border represent nodes in the set $A$; (right) the first $6$ steps of the algorithm are shown in the table;  (bottom) the activation process is shown. Activated  nodes are  shaded. At  round 0, 
 $S=\{v_1,v_5\}$.}	
\end{figure}


\noindent
{\bf Time complexity.} 
 The initialization (line 1-10) requires time $O(|V|)$.
The order in which nodes have to be considered is determined using a BFS which requires time $O(|V|)$ on a tree.
 The forall (line 11) considers all the internal nodes: the algorithm analyzes each internal node $v$ in time $O(|C(v)|)$. 
We notice that the computation in line $15$  can be executed in $O(|C(v)|)$ by using an algorithm that solve the selection problem in linear time (see for instance \cite{CLRS01}). Overall the complexity of the algorithm is $O(|V|)+ \sum_{v\in V} O(|C(v)|)= O(|V|).$

\medskip

\noindent {\bf Correctness.}
Consider the computed solution $S$. Let $\Actt{S}{0}=S$ and $\Actt{S}{i}$ be the sets of nodes which become active within  the
$i$--th round of the activation process.

\begin{lemma}\label{lemma-corr}
Algorithm {\bf $(\lambda,A)$--TSS on Tree} outputs a solution for the {\sc $(\lambda,A)$-Target Set Selection} problem  on $T=(V,E)$.
\end{lemma}

\begin{proof}
Given a node $v\in V$, let $a(v)=\min\{\lambda-\tpth{v}, \tm{v}\}$; for a leaf node  we assume
$\tpth{v}=0$.
  We prove, by induction on $a = 0, 1, \dots,$ that
 for each  $v\in A'$, s.t. $a(v) = a$ we have
\begin{equation}
	v\in \Actt{S}{a}
\end{equation}
For $a=0,$ let $v$ be a node such that $a(v)=0$. This implies
that $\tm{v}=0$  or $\lambda-\tpth{v}=0$; therefore,   $v\in S=\Actt{S}{0} $.

Now fix $a > 0$ and assume that $w \in \Actt{S}{a(w)}$ for any node $w$ with $a(w)\leq a-1$. We will prove
that $v \in \Actt{S}{a(v)}$ holds for any node $v$ with $a(v)= a$.

Let $v$ be such that $a(v) = a.$ When $v$ is processed, there are three possible cases:

\begin{sloppypar}
\begin{itemize}
\item CASE $v$ is added to the target set $S$.\\
 Actually, this case cannot occur under the standing hypothesis that $a>0$ since, if $v \in S$ then  $v\in \Actt{S}{0} $ which
 would imply $a(v)=0 < a$.
\item CASE {$(|\A{v}|\geq t(v)$)}.\\
We know that for each $u\in \A{v}$ it holds
$\tm{u}< \lambda-\tpth{v}$.
\\
In case $a=\lambda-\tpth{v}$, we have $a(u) \leq \tm{u} \leq \lambda-\tpth{v}-1=a-1.$
\\
Analogously, if $a=\tm{v}$. The algorithm poses $\tm{v}\geq \tm{u}+1$ for each $u\in \A{v}$. Therefore,
$a(u)\leq \tm{u}\leq a-1$.\\
In both the above cases  the inductive hypothesis applies to each $u\in \A{v}$, that is $\A{v}\subseteq \Actt{S}{a-1} $.
Since $|\A{v}|\geq t(v)$ we have $v\in \Actt{S}{a} $.
\item CASE $(|\A{v}|=t(v){-}1)$, $v\neq r$.\\
In such a case the algorithm sets $\tm{v}=1+ \min^{t(v)} \{\tm{u}  \ | \ u \in  C(v) \}$ where $\min^{t(v)} C$ denotes
the $t(v)$--th smallest element in the set $C$. Since $|\A{v}|=t(v){-}1$, we have that $\tm{v}>\lambda-\tpth{v}$,
hence $a(v)=\lambda-\tpth{v}$.
\\
Recalling that for each $u\in \A{v}$ it holds
$\tm{u}< \lambda-\tpth{v}$, as above we have that the inductive hypothesis applies to each $u\in \A{v}$, that is $\A{v}\subseteq \Actt{S}{a-1}$.
\\
Consider now the parent $p(v)$ of $v$.
The algorithm implies $\tpth{p(v)}\geq \tpth{v}+1$.
Hence $\lambda-\tpth{p(v)} \leq \lambda-\tpth{v}-1=a-1$ and the inductive hypothesis applies also to $p(v)$.
Therefore, $\{p(v)\}\cup \A{v}$ is a subset of size $t(v)$ of $\Actt{S}{a-1} $ and $v \in \Actt{S}{a} $.
\end{itemize}
\end{sloppypar}
We finally notice that $a(v)=\min\{\lambda-\tpth{v}, \tm{v}\}\leq \lambda$ for each $v\in A'$.
Indeed, the smallest possible value of $\pth{}$ is $-1$, which implies that $\tpth{v}\geq 0$ for any $v$.
\end{proof}

\remove{
\noindent {\it Proof (sketch).}
Given a node $v\in V$, let $a(v)=\min\{\lambda-\tpth{v}, \tm{v}\}$, where
$maxPath(v)= 1+max_{u\in C(v)} path(u)$; for definiteness we set
$\tpth{v}=0$ when $v$ is a leaf. We prove by induction on $a = 0,1,\dots,$ that
 for each  $v\in A'$, such that $a(v) = a,$ it holds that $v\in \Actt{S}{a}$.
The definition of $\tpth{v}$ and the fact that the smallest possible value of $\pth{}$ is $-1$, gives that  $\tpth{v}\geq 0$ for any $v$.
Therefore,  $a(v)=\min\{\lambda-\tpth{v}, \tm{v}\}\leq \lambda$ for each $v\in A'$, from which we have the desired result.
\hfill
}

Let $T(r)=(V,E)$ be a tree rooted at a  $r \in V$,  and let $X\subseteq V$ be a target set such that $Active[X,\lambda] \supseteq A$. Let $T(v)$ be the subtree of $T(r)$ rooted at a node $v$.
 Henceforth let  $\Acttt{X}{i}{T(v)}$ be the set of nodes that is active at round $i$ by targeting $X\cap T(v)$ in the subtree $T(v)$.
 Notice that while $X$ is a target set for $T(r)$  this not necessarily means that $X\cap T(v)$ is a target set for $T(v)$.\\

\begin{description}
	\item[--] $\tms{v}{X}=\begin{cases}
	i & \mbox { if $v \in \Acttt{X}{i}{T(v)} \setminus \Acttt{X}{i-1}{T(v)}$} \\
	0 & \mbox { if  $v \in X$} \\
	\infty & \mbox { otherwise}
	\end{cases} $
	\item[--] $\pths{v}{X}=\begin{cases}	
			0 & \mbox { if  $v \in A$ is a leaf AND $v \notin X$} \\
		i & \mbox { if $(\tpths{v}{X}=i<\lambda)$} \\
	& \mbox{ $\quad$ AND $(|\Acttt{X}{\lambda-i-1}{T(v)} \cap C(v)|= t(v)-1)$ } \\
	& \mbox{ $\quad$ AND $(v \in A$ OR $\tpths{v}{X} > 0)$}  \mbox{  AND $(v \notin X)$}  \\
	\end{cases} $
	\item[--] $\tpths{v}{X}=   \begin{cases} 1+\max_{u \in C(v)} \pths{u}{X} & \mbox { if $v$ is an internal node ($C(v)\neq \emptyset$)} \\
	0 & \mbox { if $v$ is a leaf.}
	\end{cases}$
\end{description}

%
	%
 When $X=S$ then the values $\tm{v}$ and $\pth{v}$, computed by the algorithm, correspond to the values defined above.
 %
%
%

\begin{lemma} \label{lemma:round-path-values}
If $X=S$ then for each node $v\in V,$ $\tm{v}=\tms{v}{S}$, $\pth{v}=\pths{v}{S}$ and $\A{v}=\Acttt{S}{\lambda-\tpths{v}{S}-1}{T(v)}\cap C(v)$.
\end{lemma}

\begin{proof}
First we show by induction on the height of $v$ that $\tm{v}=\tms{v}{S}$.

\noindent
{\em Induction Basis}: For each leaf $v$ we have $T(v)=v$. Since $v\notin S$, we have that $v\notin  \Acttt{S}{j}{v}$ for any value of $j$. Hence $\tms{v}{S}= \tm{v}=\infty$ (line 5).\\

\begin{sloppypar}
\noindent{\em Induction step}: Let $v$ an internal node and suppose that the claim is true for any children of $v$.\\
If $v\in S$ the claim is trivially true,  $\tms{v}{S}= 0= \tm{v}$ (lines 21 and 28).\\
Otherwise, let $\tm{v}=i=1+ \min^{t(v)} \{\tm{u}  \ | \ u \in  C(v) \}$	 where $\min^{t(v)} C$ denotes the $t(v)$--th smallest element in the set $C$.
By induction $i=1+ \min^{t(v)} \{\tms{u}{S}  \ | \ u \in  C(v) \}$. There is a set $C_S(v)\subseteq C(v)$ such that $|C_S(v)|=t(v)$, $\forall w \in C_S(v), \tms{w}{S}\leq i-1$ and $\exists u \in C_S(v)$ such that $\tms{u}{S}= i-1$. Hence,
$\forall w \in C_S(v), w \in \Acttt{S}{i-1}{T(w)} $ and $\exists u \in C_S(v)$ such that $u \in \Acttt{S}{i-1}{T(u)} \setminus \Acttt{S}{i-2}{T(u)}$ and we have $v \in \Acttt{S}{i}{T(v)}\setminus \Acttt{S}{i-1}{T(v)}$ which means that  $\tms{v}{S}=i$.
\end{sloppypar}

\medskip
Now we show that $\pth{v}=\pths{v}{S}$ and $\A{v}=\Acttt{S}{\lambda-\tpths{v}{S}-1}{T(v)}\cap C(v)$.
Again, we argue by induction on the height of $v$.

\noindent
{\em Induction Basis.} For each leaf $v$, if $v\in A$ then $\pths{v}{S}=\pth{v}=0$ (line 8). On the other hand if  $v\notin A$ then $\pth{v}=-1$ (line 10). Moreover, since $v\notin A$ and has no children $\pths{v}{S}\neq i$. Hence $\pths{v}{S}=-1.$\\
Moreover since $C(v)=\emptyset$ we have $\A{v}=\Acttt{S}{\lambda-\tpths{v}{S}-1}{T(v)}\cap C(v)=\emptyset.$

\begin{sloppypar}
\noindent {\em Induction Step.} Let $v\neq r$ be an internal node and suppose that the claim is true for any children of $v$.
Hence, $\forall u \in C(v), $ $\pth{u}=\pths{u}{S}$ and we have $\tpths{v}{S}= 1+\max_{u \in C(v)} \pth{u}=\tpth{v}.$
Notice that $\tpth{v}=1+\max_{u \in C(v)} \pths{u}{S}\geq0$.
\end{sloppypar}

We are going to show that $\A{v}=\Acttt{S}{\lambda-\tpths{v}{S}-1}{T(v)}\cap C(v)$.

Let $u \in \A{v}$. Hence $u \in C(v)$ and  $\tm{u}<\lambda-\tpth{v}$. Since $\tm{u}=\tms{u}{S}$ and by induction  $\tpth{v}=\tpths{v}{S}$ we have $\tms{u}{S}\leq \lambda-\tpth{v}-1$. Hence $u \in \Acttt{S}{\lambda-\tpth{v}-1}{T(u)}$ and therefore $u \in \Acttt{S}{\lambda-\tpth{v}-1}{T(v)}\cap C(v)$. Hence  $\A{v} \subseteq \Acttt{S}{\lambda-\tpth{v}-1}{T(v)}\cap C(v)$.

Let $u \in \Acttt{S}{\lambda-\tpth{v}-1}{T(v)}\cap C(v)$. Hence $u \in C(v)$ and $\tms{u}{S}\leq\lambda-\tpth{v}-1$.  Since $\tm{u}=\tms{u}{S}$ and by induction  $\tpth{v}=\tpths{v}{S}$ we have  $\tm{u}< \lambda-\tpth{v}$ and therefore $u \in \A{v}$. Hence  $\A{v}\supseteq \Acttt{S}{\lambda-\tpth{v}-1}{T(v)}\cap C(v)$.

\medskip

Since $v$ is an internal node, in order to show that $\pth{v}=\pths{v}{S}$  two cases have to be considered: $\pths{u}{S}=\tpth{v}$ or $\pths{u}{S}=-1$.

\begin{description}
	\item[case ($\pth{v}=\tpth{v}$):] According to the algorithm this case happens when
\begin{description}
	\item [(a)] $v\in A'$ AND
	\item [(b)]$\tpth{v}<\lambda$ AND
	\item [(c)]$|\A{v}|=t(v)-1$
\end{description}
	
Moreover, when this case occur $v$ is not added to $S$ (i.e., $v \notin S$).
Thanks to (a) we have that  $v\in A'$ if either $v\in A$ (line 2) or $v$ has a children $u$	such that $\pth{u}=\tpth{u}\geq0$ (line 23-24) or $v$ has a children $u$ such that $u \in A$ and $u$ is a leaf, that is $\pth{u}=0$ (line 7-8). Hence we have  $v \in A'$ iff  $v\in A$ OR $\tpths{v}{S} > 0.$\\
\begin{sloppypar}
Thanks to (b) and (c) we have that	$\tpths{v}{S}<\lambda$ and $|\Acttt{S}{\lambda-\tpth{v}-1}{T(v)}\cap C(v)|=t(v)-1.$
Hence using (a), (b) and (c) and the fact that  $v \notin S$ we have $\pths{v}{S}=\tpth{v}$.
\end{sloppypar}
	
	\item[case ($\pth{v}=-1$):] In this case one of the above requirement is not satisfied and we have $\pths{v}{S}=-1$.
\end{description}	

Similar reasoning can be used to show that $\pth{r}=\pths{r}{S}$ and $\A{r}=\Acttt{S}{\lambda-\tpths{r}{S}-1}{T(r)}\cap C(r)$.
\end{proof}

Let $X$ be a target set solution (i.e., $\Actt{X}{\lambda} \supseteq A$). For an edge $(v,u)$ we say that $v$ activates $u$  and write $v\to u$ if $v\in \Actt{X}{i-1}$ and $u\in \Actt{X}{i}\setminus \Actt{X}{i-1}$, for some $1\leq i\leq \lambda$.
An activation path $v\leadsto u$ from $v$ to $u$ is a path in $T$ such that $v=x_0\to x_1\to \ldots\to x_k=u$
 with $x_j\in  \Actt{X}{i_j}\setminus \Actt{X}{i_j-1}$ for $0\leq i_1< i_2<\ldots< i_k\leq \lambda.$ In other words $x_i$ is activated before $x_{i+1}$, for $i=0,\ldots, k-1$.

\begin{lemma}\label{lem:activationPath}
Let $X$  be a target set solutions (i.e., $\Actt{X}{\lambda} \supseteq A$) and $v \in V$. If $\tpths{v}{X}=i$ then there is an activation path of length $i$ in $T(v)$ starting at $v$ and ending at a node $u\in A$.
\end{lemma}

\begin{proof}{}
Since $\tpths{v}{X}=i$ then there is a path in $T(v)$ from $v$ to a node $u$ such that $v=x_i\to x_{i-1}\to \ldots\to x_0=u$ where for each $i=0,1,\ldots, i-1,$ $\pths{x_i}{X}=i.$
We are able to show by induction that for each $j=0,1,\ldots ,i-1,$ $x_j$ is activated after $x_{j+1}$.

\noindent
{\em Induction basis}: $j=0$. Hence, $\pths{x_0}{X}=0$ which means that $x_0\notin X$. There are two case to consider:

\begin{description}
	\item[($x_0$ is a leaf)] hence $x_0 \in A$. Moreover since $x_0$ has no children we have that $x_0$ will be activated after its parent $x_{1}$.
	\item[($x_0$ is an internal node)] since  $\pths{x_0}{X}=0$ we have that $\tpths{x_0}{X}=0$ hence $x_0 \in A$. Moreover, since $|\Acttt{X}{\lambda-1}{T(x_0)} \cap C(x_0)|= t(x_0)-1$ we have that $x_0$ will be activated after its parent $x_{1}$. Otherwise $x_0$ will not be activated by round $\lambda-1$.
\end{description}

\noindent
{\em Induction step}: $j=i$. Hence, $\pths{x_j}{X}=\tpths{x_j}{X}=j$ which means that $x_i\notin X$. Moreover, by induction,  we know that $\forall j<i$, $x_j$ is activated after $x_{j+1}$. Hence in order to activate $x_0$ by round $\lambda$, $x_j$ has to be activated by round $\lambda-j-1$. Since $|\Acttt{X}{\lambda-j-1}{T(x_j)} \cap C(x_j)|= t(x_j)-1$ we have that $x_j$ will be activated after its parent $x_{j+1}$. Otherwise $x_j$ will not be activated by round $\lambda-j-1$.
\end{proof}

Let $X$  and $Y$ be two target set solutions and $v \in V$. The following properties hold:
	
\begin{proper} 	\label{p1}
If $\tms{v}{X}>\tms{v}{Y}$ and $v \notin Y$ then there exists $u \in C(v)$ such that $\tms{u}{X}>\tms{u}{Y}$.
\end{proper}

\begin{proof}{}
Let $\tms{v}{Y}=i$ we have that $v \in \Acttt{Y}{i}{T(v)} \setminus \Acttt{Y}{i-1}{T(v)}$. Hence, there is a set $C_Y\subseteq C(v)$  such that $|C_Y|=t(v)$ and $\forall u \in C_Y, u \in  \Acttt{Y}{i-1}{T(u)}$  (i.e., $\forall u \in C_Y, \tms{u}{Y}\leq i-1$). On the other hand, since  $\tms{v}{X}>i$ the size of the set $C_X=\{ u \in C(v) | \tms{u}{X} \leq i-1 \} $ is at most $t(v)-1$. Hence there is at least a vertex $u \in C(v)$ such that $\tms{u}{X}>\tms{u}{Y}$.
\end{proof}

\begin{proper}\label{p3}
If $v \notin X$ then $\tpths{v}{X}<\lambda$ AND $|\Acttt{X}{\lambda-\tpths{v}{X}-1}{T(v)} \cap C(v)|> t(v)-2$.
\end{proper}

\begin{sloppypar}
\begin{proof}{}
In the following we show that if either  $\tpths{v}{X}\geq\lambda$ or $|\Acttt{X}{\lambda-\tpths{v}{X}-1}{T(v)} \cap C(v)|\leq t(v)-2$ then $v \in X.$

When $\tpths{v}{X}\geq\lambda$  then by Lemma \ref{lem:activationPath} there is an activation path of length at least $\lambda$ in $T(v)$ starting at $v$ and
ending at a node $u \in A$ and  we have that $v$ has to be active at round $0$.

When $|\Acttt{X}{\lambda-\tpths{v}{X}-1}{T(v)} \cap C(v)|\leq t(v)-2$, since by Lemma \ref{lem:activationPath} there is an activation path of length $\tpths{v}{X}$ starting at $v$ and ending at a node $u \in A$, we have that $v$ has to be active at round $\lambda-\tpths{v}{X}$ (i.e., $v$ should belong to $\Actt{X}{\lambda-\tpths{v}{X}}$). Since $|\Acttt{X}{\lambda-\tpths{v}{X}-1}{T(v)} \cap C(v)|\leq t(v)-2$ then $v$ will not be activated (even considering its parent) at round $\lambda-\tpths{v}{X}$. Hence $v$ has to be in $X$.
\end{proof}	
\end{sloppypar}
	
\begin{lemma}\label{lemma-opt}
Algorithm \textbf{$(\lambda,A)-$TSS on Tree} outputs an optimal solution for the {\sc $(\lambda,A)-$ Target Set Selection} problem  on $T=(V,E)$.
\end{lemma}

\begin{proof}
Let $S$ and $O$ be respectively the solutions found  by the Algorithm \textbf{$(\lambda,A)-$TSS on Tree} and an optimal solution.
For each $v \in V$  let $S(v) = S \cap T(v)$ (resp. $O(v) = O \cap T(v)$) be the set of target nodes in $S$ (resp. $O$) which belong to $T(v)$.
Let $s(u)=|S(u)|$ and $o(u)=|O(u)|$ be the cardinality of such sets. We will use the following claim.

\begin{claim}
For any vertex $v \in V$, \textbf{if} $\pths{v}{S}>\pths{v}{O}$  \textbf{ OR }   $\tms{v}{S}>\tms{v}{O}$ \textbf{then} $s(v) < o(v).$
\end{claim}

\begin{proof} We argue by induction on the height of $v.$
The claims trivially hold when $v$ is a leaf. Since our algorithm does not target any leaf (i.e. $v \notin S$), two cases need to be analyzed:
\begin{description}
	\item[case ($v \in O$):] then $s(v)=0<1=o(v)$ and  the inequality is satisfied.
	\item[case ($v \notin O$):]  then $\tms{v}{S}=\tms{v}{O}=\infty$ and $\pths{v}{S}=\pths{v}{O}=0\mbox{ or } {-}1$ depending whether $v\in A$ or not. Hence none of the two conditions of the {\bf if} are satisfied then the claim holds.
\end{description}		

\begin{sloppypar}
Now consider any internal vertex $v \in V$. By induction,  we have that  $ \forall u\in C(v), s(u)\leq o(u)$, hence
	 \begin{equation}
	 \sum_{u \in C(v)} s(u)\leq \sum_{u \in C(v)} o(u). \label{eqSO}
	 \end{equation}
	  There are four cases to consider:
\end{sloppypar}
\begin{description}

	\item[case ($v \notin S$ and $v \in O$):] Using eq. (\ref{eqSO}) we have $s(v)\leq o(v)-1<o(v),$ hence  the claim holds.
	
\begin{sloppypar}
	
	\item[case ($v \in S$ and $v \in O$):]  We have   $\pths{v}{S}=\pths{v}{O}=-1$ and $\tms{v}{S}=\tms{v}{O}=0$. Hence none of the two conditions of the
	{\bf if} statement are satisfied and the claim holds.

\end{sloppypar}

	\item[case  ($v \notin S$ and $v \notin O$):]
	
\textbf{if} $\pths{v}{S}>\pths{v}{O}$  \textbf{ OR } $\tms{v}{S}>\tms{v}{O}$  then in order to prove the claim  we need to find a child $u$ of $v$ such that $s(u)<o(u)$.

	In the following we analyze the two cases separately:
	
if  $\tms{v}{S}>\tms{v}{O}$ then  using Property  \ref{p1} we have that there is a vertex $u \in C(v)$ such that $\tms{u}{S}> \tms{u}{O}$. By induction on the height of $v$  we have that $s(u) < o(u)$.

if $\pths{v}{S}>\pths{v}{O}$ then $\pths{v}{S}\neq -1$, hence by definition of $\pths{\cdot}{X}$,
\begin{eqnarray}
\tpths{v}{S}=i<\lambda\\
\mbox{ AND } |\Acttt{S}{\lambda-i-1}{T(v)} \cap C(v)|= t(v)-1 \label{eqAct} \\
\mbox{ AND } (v \in A \mbox{ OR } \tpths{v}{S} > 0) \label{eqA}
\end{eqnarray}
		There are two cases to consider:
\begin{sloppypar}
		\begin{description}
		\item[case ($\pths{v}{O}\geq 0$):] then $\pths{v}{O}=	\tpths{v}{O} < \tpths{v}{S}$ and there is a child $u$ of $v$ such that $\pths{u}{S} = \tpths{v}{S}-1> \tpths{v}{O}-1 \geq \pths{u}{O}$ and we have found the desired vertex because by induction  we have $s(u)<o(u)$.	
			\item[case ($\pths{v}{O}=-1$):] Since $v \notin O$ and by Property \ref{p3} we know that
			$\tpths{v}{O}<\lambda$ AND $|\Acttt{O}{\lambda-\tpths{v}{O}-1}{T(v)} \cap C(v)|> t(v)-2$. Hence $\pths{v}{O}=-1$  can happen only for two reasons:
			\begin{description}
			\item[case ($v \notin A$ and  $\tpths{v}{O}=0$):] Since $v \notin A$, then by eq. \ref{eqA},  there is a children $u$ of $v$ such that $\pths{u}{S}\geq0>-1=\pths{u}{O}$  and we have found the desired vertex because by induction it holds that  $s(u)<o(u)$.
		\item[case] ($|\Acttt{O}{\lambda-\tpths{v}{O}-1}{T(v)} \cap C(v)| \geq t(v)$):
		Hence by equation \ref{eqAct} we have   $$|\Acttt{O}{\lambda-j-1}{T(v)} \cap C(v)|    > |\Acttt{S}{\lambda-i-1}{T(v)} \cap C(v)|$$ where $i=\tpths{v}{S}$ and $j=\tpths{v}{O}$.
		If $i=\tpths{v}{S}>\tpths{v}{O}=j$ (i.e., $1+\max_{u \in C(v)} \pths{u}{S}>1+\max_{u \in C(v)} \pths{u}{O}$) then there is a child $u$ of $v$ such that $\pths{u}{S}>\pths{u}{O}$ and we have found the desired vertex because by induction we have $s(u)<o(u)$. On the other hand if $i=\tpths{v}{S}\leq\tpths{v}{O}=j$, then there exists  a child $u \in C(v)$  such that, $u \in \Acttt{O}{\lambda-j-1}{T(u)}\setminus \Acttt{O}{\lambda-j-2}{T(u)} $ and  $u \notin \Acttt{S}{\lambda-i-1}{T(u)}$. Hence $\tms{u}{O}=\lambda-j-1 \leq \lambda-i-1<\tms{u}{S}$.   By induction we have that $s(u) < o(u)$.
\end{description}
\end{description}
\end{sloppypar}
	In all the cases above we are able to find the desired vertex and the claim  holds.

	\item[case ($v \in S$ and $v \notin O$):] Using eq. (\ref{eqSO}) we have $s(v)-1\leq o(v).$ For each $u\in C(v)$ we know by induction that $s(u)\leq o(u)$. Since $v \in S$ we have $\pths{v}{S}=-1$ and $\tms{v}{S}=0$, hence none of the two requirement of the {\bf if}  is satisfied hence the claim holds true.
	\end{description}
\end{proof}

We show  by induction on the height of the node  that $s(v) \leq o(v)$, for each  $v \in V$.\\
The inequality trivially holds when $v$ is a leaf. Since our algorithm does not target any leaf (i.e. $v \notin S$), we have
$s(v) = 0$. Since we have $o(v)= 0$ or $o(v) = 1$ according to whether $v$ belongs to the optimal solution,   the inequality is always satisfied.
\\
Now consider any internal node $v $. By induction,   $s(u)\leq o(u)$ for each $u\in C(v)$; hence
	 \begin{equation}
	 \sum_{u \in C(v)} s(u)\leq \sum_{u \in C(v)} o(u). \label{eqSO-body}
	 \end{equation}
It is not hard to see that if $v \in O$ by  (\ref{eqSO-body}) we immediately have
$s(v) \leq o(v).$ The same result follows from (\ref{eqSO-body}) for the case where $v$ is neither in $O$ nor in $S.$
\\
We are left with the case $v  \in S$ and $v \not \in O$
In this case eq. (\ref{eqSO-body}) only gives $s(v)-1\leq o(v).$
In order to obtain the desired result we need to find a child $u$ of $v$ such that $s(u)<o(u)$.
We distinguish the following two cases:\\
\noindent	\textbf{case ($v \notin A$ and  $\tpths{v}{O}=0$):}   Since $v \in S$ we have $v \in A'$ (that is $v \in A$ or $\tpths{v}{S}>0$). Since $v \in A'\setminus A$ then $\tpths{v}{S}>0$ and  there is a children $u$ of $v$ such that $\pths{u}{S}\geq0>-1= \pths{u}{O}= -1$  and we have found the desired vertex because by the Claim above we have $s(u)<o(u)$.\\
\noindent		\textbf{case ($v \in A$ or  $\tpths{v}{O}>0$):}
Since $v \in S$ we have that either $\tpths{v}{S} = \lambda$ or $|\A{v}|=|\Acttt{S}{\lambda-i-1}{T(v)} \cap C(v)| \leq t(v)-2$ where $i=\tpths{v}{S}$. We consider the two subcases separately:

\begin{description}
		\item [--] ($\tpths{v}{S} = \lambda$):
		 Since $v \notin O$, by Property \ref{p3},  we have that $\tpths{v}{O} < \lambda$. Hence,  there is a vertex $u \in C(v)$ such that $\pths{u}{S}=\lambda-1>\pths{u}{O}$ and we have found the desired vertex because, by the Claim we have $o(u)>s(u)$.
		\item [--] ($|\A{v}|=|\Acttt{S}{\lambda-\tpths{v}{S}-1}{T(v)} \cap C(v)| \leq t(v)-2$):
		Since $v \notin O$ by Property \ref{p3}   we have that $|\Acttt{O}{\lambda-\tpths{v}{O}-1}{T(v)} \cap C(v)|> t(v)-2$.
Hence, $|\Acttt{O}{\lambda-j-1}{T(v)} \cap C(v)|    > |\Acttt{S}{\lambda-i-1}{T(v)} \cap C(v)|$
 where $i=\tpths{v}{S}$ and $j=\tpths{v}{O}$. 		
If $i=\tpths{v}{S}>\tpths{v}{O}=j$ (i.e., $1+\max_{u \in C(v)} \pths{u}{S}>1+\max_{u \in C(v)} \pths{u}{O}$) then there is a child $u$ of $v$ such that $\pths{u}{S}>\pths{u}{O}$ and we have found the desired vertex becaus, by the Claim we have $o(u)>s(u)$. On the other hand if $i=j$, then there exists  a child $u \in C(v)$  such that, $u \in \Acttt{O}{\lambda-j-1}{T(u)}\setminus \Acttt{O}{\lambda-j-2}{T(u)} $ and  $u \notin \Acttt{S}{\lambda-i-1}{T(u)}$. Hence $\tms{u}{O}=\lambda-j-1 \leq \lambda-i-1<\tms{u}{S}$.   By the Claim we have $o(u)>s(u)$.
 		\end{description}
In all cases we have that  there is   $u \in C(v)$ with $s(u) < o(u)$. Hence $s(v) \leq o(v).$
\end{proof}

\bibliographystyle{plain}

\end{document}